\documentclass[12pt]{article}
\usepackage{amsthm,amsmath,amssymb}

\usepackage{caption}
\usepackage{booktabs}
\usepackage{tabularx}
\usepackage{multirow}
\usepackage{times}
\usepackage[colorlinks]{hyperref}
\usepackage{graphicx} 
\usepackage{algorithmicx,algorithm}
\usepackage[noend]{algpseudocode}

 \usepackage{geometry}
\newcommand{\prealg}[2]{
	\begin{tabular}{l p{0.85\textwidth}}
		{\textbf{Input:}} & #1.\\
		{\textbf{Output:}}&  #2.\\
	\end{tabular}
}

\newtheorem{theorem}{Theorem}
\newtheorem{definition}{Definition}
\newtheorem{property}{Property}
\newtheorem{claim}{Claim}

\newtheorem{lemma}{Lemma}

\newcommand{\yesins}{Yes-instance}

\newcommand{\noins}{No-instance}

\newcommand{\np}{{NP}}
\newcommand{\nph}{{\np}-hard}

\newcommand{\bigo}[1]{O(#1)}
\newcommand{\abs}[1]{|#1|}
\newcommand{\edge}[2]{#1#2}
\newcommand{\prob}[1]{{\sc{#1}}}
\newcommand{\fstcpri}[2]{K_{#1}^{#2}} 
\newcommand{\vsetprim}[2]{V(#1, #2)}
\newcommand{\rsto}[2]{#1|_{#2}} 
\newcommand{\parent}[1]{{\sf{parent}}(#1)}
\newcommand{\oover}{{on}} 
\newcommand{\ii}{i}
\newcommand{\vv}{\mathcal{V}}
\newcommand{\dst}[1]{{\sf{dst}}(#1)}
\newcommand{\cdst}[1]{{\sf{dst}}[#1]}
\newcommand{\mneg}[1]{{#1}}
\newcommand{\onlyfull}[1]{}

\newcommand{\pp}{p}

\newcommand{\EPP}[3]
{\begin{center}
{\small
\begin{tabularx}{0.98\columnwidth}{ll}
\toprule
\multicolumn{2}{l}{\textsc{#1}} \\ \midrule
{\bf Given:}   & \parbox[t]{0.8\columnwidth}{#2\vspace*{1mm}}  \\
{\bf Question:}& \parbox[t]{0.8\columnwidth}{#3\vspace*{.5mm}} \\ \bottomrule
\end{tabularx}
}
\end{center}}

\begin{document}
	
		
\title{A Polynomial-Time Algorithm for MCS Partial Search Order on Chordal Graphs\thanks{A preliminary version of the paper appeared in the Proceedings of the 48th International Symposium on Mathematical Foundations of Computer Science (MFCS~2023)~\protect\cite{DBLP:conf/mfcs/ZhenYL23}.}}


\author{Guozhen Rong$^1$, Yongjie Yang$^2$, Wenjun Li$^{1}$\thanks{Corresponding author.}}
\date{\small{$^1$Hunan Provincial Key Laboratory of Intelligent Processing of Big Data on Transportation,\\ Changsha University of Science and Technology, Changsha, China\\
\{rongguozhen, lwjcsust\}@csust.edu.cn\\
$^2$Chair of Economic Theory, Saarland University, Saarbr\"{u}cken, Germany\\ yyongjiecs@gmail.com}}

\maketitle

\begin{abstract}
We study the partial search order problem (\prob{PSOP}) proposed recently by Scheffler [WG 2022]. Given a graph~$G$ together with a partial order {\oover} the set of vertices of~$G$, this problem determines if there is an~$\mathcal{S}$-ordering that is consistent with the given partial order, where~$\mathcal{S}$ is a graph search paradigm like BFS, DFS, etc. This problem naturally generalizes the end-vertex problem which has received much attention over the past few years. 
It also generalizes the so-called ${\mathcal{F}}$-tree recognition problem which has just been studied in the literature recently. 
Our main contribution is a polynomial-time dynamic programming algorithm for the {\prob{PSOP}} of the maximum cardinality search (MCS) restricted to chordal graphs. This resolves one of the most intriguing open questions left in the work of Scheffler [WG 2022]. 
To obtain our result, we propose the notion of layer structure and study numerous related structural properties which might be of independent interest. 
\smallskip

\end{abstract}
		
		

\section{Introduction}

Graph search paradigms are pervasive in algorithms for innumerable graph problems. 
In addition to the most popular paradigms breadth-first search (BFS) and depth-first search (DFS), several other prevalent graph search paradigms---including, for instance, lexicographic breadth-first search (LBFS), lexicographic depth-first search (LDFS), maximum cardinality search (MCS), maximal neighborhood search (MNS)---have also been extensively studied in the literature~\cite{DBLP:journals/siamdm/CorneilK08,DBLP:journals/siamcomp/RoseTL76,DBLP:journals/siamcomp/Tarjan72,DBLP:journals/siamcomp/TarjanY84}. 
These graph search paradigms have proved to be exclusively useful in dealing with a variety of graph problems~\cite{DBLP:journals/siamdm/BretscherCHP08,DBLP:journals/siamcomp/CorneilDH13,DBLP:journals/siamdm/CorneilDHK16,DBLP:journals/dam/KumarM98}. For instance, MCS has been successfully used in the recognition of special graphs~\cite{DBLP:journals/siamcomp/TarjanY84}, the computation of minimal separators~\cite{DBLP:journals/dam/KumarM98}, the computation of minimal triangulation of graphs~\cite{DBLP:journals/algorithmica/BerryBHP04}, determining lower bounds of treewidth~\cite{DBLP:journals/dam/BodlaenderK07,DBLP:journals/endm/KosterBH01}, etc.
In several of these algorithmic applications, last visited vertices in graphs are crucial for the correctness of the algorithms. Last visited vertices  also exhibit some nice structural properties. For instance, for a cocomparability graph that is Hamiltonian, if a vertex is last visited by LDFS, then there is a Hamiltonian path starting from this vertex~\cite{DBLP:journals/dmtcs/BeisegelDKKPSS19,DBLP:journals/siamcomp/CorneilDH13}. For more concrete examples on this issue, we refer to~\cite{DBLP:journals/dmtcs/BeisegelDKKPSS19,DBLP:journals/dam/CorneilKL10}. These specialities of last visited vertices inspired Corneil, K\"{o}hler, and Lanlignel~\cite{DBLP:journals/dam/CorneilKL10} to put forward the end-vertex problem, in which we are given a graph and a particular vertex~$v$, and are asked whether~$v$ can be the last visited one according to a certain graph search paradigm. Thenceforth, investigation on the end-vertex problem has flourished, resulting in the complexity of the problem for both general graphs and many special graphs such as chordal graphs, split graphs, interval graphs, bipartite graphs, etc., being substantially  established~\cite{DBLP:journals/dmtcs/BeisegelDKKPSS19,RONG22,DBLP:journals/ipl/ZouWWC22}. For a summary of the recent progress, we refer to~\cite{Gorznyphdthesis2022endvertex}. 

A closely related problem is the search tree recognition problem which has a relatively longer history~\cite{HagerupNowak1985b,HagerupNowak1985,DBLP:conf/wg/KorachO88}. 
This problem determines if a given spanning tree of a graph can be obtained via a traversal of the graph by a certain search paradigm. This problem comes natural for some search paradigms like BFS and DFS, since they not only output an ordering but also generate a spanning tree during the search. However, it is ill-defined for some other search paradigms like MCS and MNS. Aiming at overcoming the plight, Beisegel~et~al.~\cite{BeisegelDKKPSS21} introduced the notions of $\mathcal{F}$-tree and  $\mathcal{L}$-tree ($\mathcal{F}$ and $\mathcal{L}$ respectively stand for ``first'' and ``last''). Particularly, given an ordering~$\sigma$ of the vertices of a graph, the $\mathcal{F}$-tree (respectively, $\mathcal{L}$-tree) is a spanning tree of the graph so that every vertex~$v$ other than the first one in~$\sigma$ is adjacent to its first (respectively, last) neighbor appearing before~$v$ in~$\sigma$. BFS-trees and DFS-trees are $\mathcal{F}$-trees and $\mathcal{L}$-trees of BFS and DFS, respectively. 
Having these notions, Beisegel~et~al.~\cite{BeisegelDKKPSS21} studied the complexity of the $\mathcal{F}$-tree recognition problem and the $\mathcal{L}$-tree recognition problem with respect to the above-mentioned search paradigms for both general graphs and many special graph classes. Very recently, Scheffler~\cite{DBLP:journals/tcs/Scheffler22} complemented these results by showing that the {$\mathcal{L}$}-tree recognition problem of BFS restricted to bipartite graphs, and the {$\mathcal{F}$}-tree recognition problem of DFS restricted to chordal graphs and chordal bipartite graphs are {\nph}, standing in a strong contrast to the polynomial-time solvability of the {$\mathcal{F}$}-tree recognition problem of BFS and the {$\mathcal{L}$}-tree recognition problem of DFS in general~\cite{HagerupNowak1985,DBLP:conf/wg/KorachO88}. 

From the known results, Scheffler~\cite{DBLP:conf/wg/Scheffler22} discerned that  the complexity of the end-vertex problem and the complexity of the $\mathcal{F}$-tree recognition problem seemed to be somewhat connected. For instance, for LBFS, MCS, and MNS, both problems are {\nph} on weakly chordal graphs and are linear-time solvable on split graphs. Additionally, for MNS and MCS, both problems are polynomial-time solvable when restricted to chordal graphs~\cite{BeisegelDKKPSS21,RONG22}.  
Towards a comprehensive understanding of the connection, Scheffler~\cite{DBLP:conf/wg/Scheffler22} introduced the partial search order problem ({\prob{PSOP}}) which generalizes both the end-vertex problem and the {$\mathcal{F}$-tree recognition problem}. Given a graph~$G$ and a partial order~$R$ {\oover} the vertices of~$G$, the {\prob{PSOP}} of a search paradigm~$\mathcal{S}$ determines if~$G$ admits an $\mathcal{S}$-ordering which linearly extends~$R$.  Scheffler~\cite{DBLP:conf/wg/Scheffler22} derived polynomial-time algorithms for the {\prob{PSOP}} of LBFS restricted to chordal bipartite graphs, and polynomial-time algorithms for the {\prob{PSOP}} of MCS restricted to split graphs. However, whether the {\prob{PSOP}} of MCS restricted to chordal graphs, arguably the most intriguing case, is polynomial-time solvable is unknown prior to our current work. 
We resolve this open question in the affirmative.  To obtain our result, we propose the notion of layer structure and study a number of structural properties which might be of independent interest. At a high level, based on the properties studied, we iteratively decompose the clique graph of a given chordal graph into what we call layer structures, handle the components (which we call units) of each layer structure separately, and utilize dynamic programming technique to merge local solutions into a whole one.

\section{Preliminaries}
In this section, we elaborate on important terminologies and  notions used to obtain our results. By convention,~$[i]$ denotes the set of positive integers no greater than~$i$. 

\subsection{Graphs}
We follow standard notions in graph theory. For notions used but not defined in the paper, we refer to~\cite{Douglas2000}.
The graphs we consider are finite, undirected, and simple. Let~$G$ be a graph. The {\emph{vertex set}} and {\emph{edge set}} of~$G$ are denoted by~$V(G)$ and~$E(G)$, respectively. 
For a vertex $v \in V(G)$, its \emph{neighborhood} in~$G$, denoted~$N_G(v)$, is the set of vertices adjacent to~$v$ in~$G$.  A \emph{clique} of~$G$ is a subset of pairwise adjacent vertices in~$G$. We call a clique of~$G$ containing a vertex $v\in V(G)$ a {\emph{$v$-clique}}. Analogously, a clique of~$G$ containing a subset $X\subseteq V(G)$ is called an {\emph{$X$-clique}}. The subgraph of~$G$ induced by $X\subseteq V(G)$ is denoted by~$G[X]$. 

A path~$P$ of length~$t$ is a graph with a sequence of $t+1$ distinct vertices~$v_1$,~$v_2$, $\dots$, $v_{t+1}$ and with the edge set $\{\edge{v_i}{v_{i+1}} : i\in [t]\}$. We say that~$P$ is a path between~$v_1$ and~$v_t$, or simply call it a $v_1$-$v_t$ path. 
Two vertices in~$G$ are {\emph{connected}} if there is a path between them. For $u, v\in V(G)$, a {\emph{$u$-$v$ separator}} is a subset $X\subseteq V(G)$ so that~$u$ and~$v$ are disconnected after deleting all vertices in~$X$ from~$G$. A $u$-$v$ separator~$X$ is {\emph{minimal}} if there are no other $u$-$v$ separators~$X'$ such that $X'\subsetneq X$.  
 
The {\emph{length}} of a cycle is the number of edges it contains. A \emph{hole} is an induced cycle of length greater than three.  A graph is \emph{chordal} if it does not contain any holes as induced subgraphs.

\subsection{The Partial Search Order Problem}

A {\emph{partial order}} {\oover} a set~$X$ is a reflexive, antisymmetric, and transitive binary relation {\oover}~$X$. For ease of exposition, for a partial order~$R$, we sometimes use $x <_R y$ to denote $(x, y)\in R$. 
A {\emph{linear order}} is a partial order that is complete. We usually write a linear order~$R$ in the format of $(x_1, x_2, \dots, x_m)$ which means that $(x_i, x_j)\in R$ for all $i, j\in [m]$ such that $i<j$. 
 A linear order~$R$ {\emph{extends}} a partial order~$R'$ if for every $(x, y)\in R'$ it holds that $(x, y)\in R$. We also call~$R$ a {\emph{(linear) extension}} of~$R'$. 
For a binary relation~$R$ {\oover} a set~$X$, and for~$X'\subseteq X$, we use~$\rsto{R}{X'}$ to denote~$R$ restricted to~$X'$. 
 
For a graph search paradigm~$\mathcal{S}$ and a graph~$G$, an {\emph{$\mathcal{S}$-ordering}} of~$G$ is an ordering of~$V(G)$ that can be generated from an~$\mathcal{S}$ search on~$G$. 

The partial search order problem ({\prob{PSOP}}) of a graph search paradigm~$\mathcal{S}$ is defined as follows. 

\EPP 
{PSOP-$\mathcal{S}$}
{A connected graph~$G$ and a partial order~$R$ {\oover}~$V(G)$.}
{Is there an $\mathcal{S}$-ordering of~$G$ that extends~$R$?}

In the paper, we study the case where~$\mathcal{S}$ is MCS and the input graph is a connected chordal graph.

\subsection{Clique Graphs}

It has long been known that chordal graphs admit a characterization in terms of their clique trees. 
Precisely, a connected graph~$G$ is chordal if and only if there exists a tree~$T$ whose vertices one-to-one correspond to maximal cliques of~$G$ so that for every vertex $v \in V(G)$ the vertices of~$T$ corresponding to all maximal $v$-cliques of~$G$ induce a subtree of~$T$~\cite{blair1993,DBLP:journals/dm/Buneman74,Gavril1974,DBLP:journals/siamcomp/TarjanY84}. 
Such a tree~$T$ is referred to as a \emph{clique tree} of~$G$~\cite{DBLP:journals/dm/Buneman74,Gavril1974,WALTERphd1972}.  

\begin{figure}[htbp]
	\centering
\includegraphics[width=0.85\textwidth]{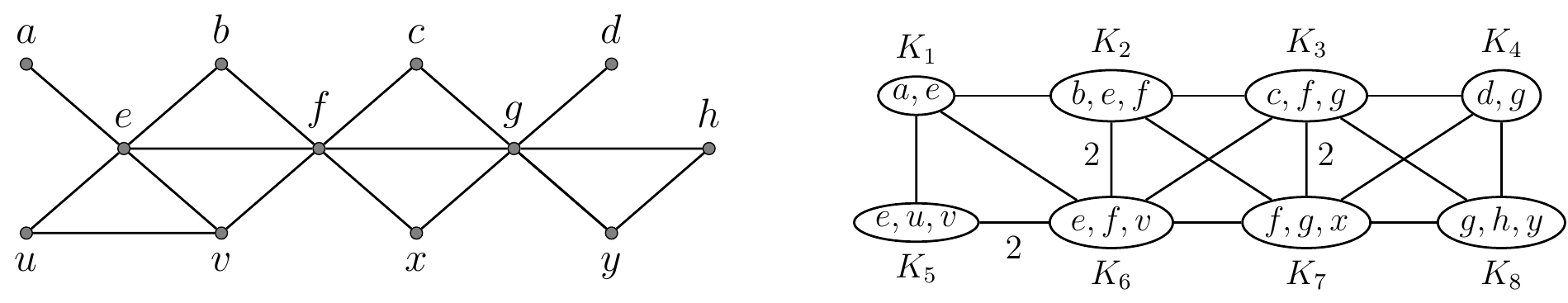}
	\caption{A connected chordal graph (left) and its clique graph (right). In the clique graph, all omitted edge weights are~$1$.}
	\label{fig-clique-graph-example}
\end{figure}

Another relevant notion is \emph{clique graph}, first introduced by Galinier, Habib, and Paul ~\cite{GalinierHP95}.\footnote{This notion is also studied under the name reduced clique graph in the literature (see, e.g.,~\cite{DBLP:journals/ejc/HabibS12}).} Precisely, the clique graph of a connected chordal graph~$G$, denoted~$C(G)$, is the graph whose vertex set is exactly the set of all maximal cliques of~$G$, and two vertices~$K$ and~$K'$ in the clique graph are adjacent if and only if $K \cap K'$ is a minimal $u$-$v$ separator of~$G$ for all $u \in K\setminus K'$ and all $v \in K'\setminus K$. Each edge~$\edge{K}{K'}$  in the clique graph~$C(G)$ is associated with the {\emph{label}} $K \cap K'$ and with the {\emph{weight}} $|K \cap K'|$.

For clarity, hereinafter we call vertices in a clique tree or a clique graph {\emph{nodes}}.  
It is a folklore that every chordal graph~$G$ has at most~$|V(G)|$ maximal cliques~\cite{Dirac1961}, and hence every clique tree/graph of~$G$ contains at most~$|V(G)|$ nodes.  
Clique trees and clique graphs  of chordal graphs are closely linked, as stated in the following lemma.

\begin{lemma}[\cite{GalinierHP95}]
\label{lem-clique-graph}
	Let~$G$ be a connected chordal graph~$G$, and let~$K$ and~$K'$ be two maximal cliques in~$G$. Then,~$K$ and~$K'$ are adjacent in~$C(G)$ if and only if they are adjacent in some clique tree of~$G$. 
\end{lemma}

In effect, Lemma~\ref{lem-clique-graph} asserts that the clique graph of a connected chordal graph is the union of all clique trees of the same graph. 
As an important consequence, it holds that for every $v\in V(G)$, the subgraph of~$C(G)$ induced by all maximal $v$-cliques are connected~\cite{DBLP:journals/ejc/HabibS12}.

For a label~$S$ of some edge in~$C(G)$, we use~$C(G) \ominus S$ to denote the graph obtained from~$C(G)$ by deleting all edges with the label~$S$. For example, for the chordal graph $G$ in Figure~\ref{fig-clique-graph-example}, $C(G)\ominus \{f\}$ is the graph shown in Figure~\ref{fig-clique-graph-example-b}.

\begin{figure}[htbp]
	\centering
	\includegraphics[width=0.4\textwidth]{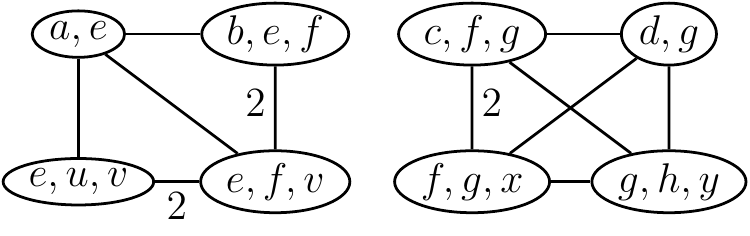}
	\caption{$C(G)\ominus \{f\}$ where $G$ is the connected chordal graph in Figure~\ref{fig-clique-graph-example}.}
	\label{fig-clique-graph-example-b}
\end{figure}

\begin{lemma}[\cite{RONG22}]
\label{lem-clique-graph-delete-edge}
	Let~$G$ be a connected chordal graph, and let~$S$ be the label of an edge in~$C(G)$ with the minimum weight. Then, 
	\begin{enumerate}
		\item[(i)] for every  $v \in V(G) \setminus S$, all maximal $v$-cliques of~$G$ are in the same connected component of $C(G) \ominus S$; and
		\item[(ii)] for every $u,v \in V(G) \setminus S$,~$S$ is a~$u$-$v$ separator in~$G$ if and only if any maximal $u$-clique  and any maximal $v$-clique of~$G$ are in different connected components of~$C(G) \ominus S$.
	\end{enumerate}
\end{lemma}

\subsection{Graph Search Paradigms}
Now we give definitions of three graph search paradigms, namely,  MCS, the Prim search, and the generic search. Our focus is {\prob{PSOP-MCS}}, but our algorithm resorts to Prim search of clique graphs and generic search of layer structures (defined in Section~\ref{sec-ls}) of clique graphs. 

Let us start with MCS. Simply put, beginning with an arbitrary vertex,  MCS picks a vertex having the maximum number of visited neighbors as the next to visit. See Algorithm~\ref{alg-MCS} for a pseudocode of MCS. 

\begin{algorithm}[htbp]
	\caption{MCS}
	\label{alg-MCS}
	\prealg
	{A connected graph~$G$ of~$n$ vertices}
	{An ordering $\sigma$ of~$V(G)$}
	\begin{algorithmic}[1]
		\For{each $v \in V(G)$} 
		\State ${\sf{count}}(v) \leftarrow 0$;
		\EndFor
		\For{$i = 1, 2, \dots, n$} 
		\State let~$v$ be any unvisited vertex so that ${\sf{count}}(v)\geq {\sf{count}}(v')$ for all unvisited vertices~$v'$;
		\State $\sigma(v) \leftarrow i$; {\hfill /* visit~$v$ */}
		\For{each unvisited vertex $u \in N_G(v)$} 
		\State ${\sf{count}}(u) \leftarrow {\sf{count}}(u)+1$;
		\EndFor
		\EndFor
		\State \Return $\sigma$;
	\end{algorithmic}
\end{algorithm}

Prim's algorithm is one of the most famous algorithms for finding minimum spanning trees~\cite{Dijkstra1959,Jarnik1930,Prim1957}. In a nutshell, starting from a tree consisting of an arbitrary edge with the minimum weight, the algorithm grows the tree by adding edges, one-by-one, of minimum possible weights without destroying the tree structure, until the tree becomes a spanning tree. By turning ``minimum'' to ``maximum'' in the algorithm, it instead  returns a maximum spanning tree.     
Algorithm~\ref{alg-Prim-Search} delineates a graph search paradigm modified from Prim's algorithm~\cite{RONG22}. Generally speaking, starting from any arbitrary vertex, it picks as the next one to be visited a so far unvisited vertex incident to an edge with the maximum weight among all edges between visited vertices and unvisited vertices. 
Following~\cite{RONG22}, we call an ordering obtained from applying Algorithm~\ref{alg-Prim-Search} to a graph~$G$ a \emph{Prim ordering} of~$G$.  
Prim orderings of the clique graph of a chordal graph have an appealing property in respect of their clique graphs, as stated in the following lemma. 

\begin{algorithm}[htbp]
	\caption{Prim Search}
	\label{alg-Prim-Search}
	\prealg
	{A connected graph~$G$ of~$n$ vertices where every edge has a weight}
	{An ordering $\pi$ of~$V(G)$}
	\begin{algorithmic}[1]
		\State $\pi(v) \leftarrow 1$, where~$v$ is an arbitrary vertex of~$G$;  {\hfill /* the first visited vertex */}
		\State $S\leftarrow \{v\}$;
		\For{$i=2, 3, \dots, n$} 
		\State let~$v$ be a vertex in $V(G)\setminus S$  incident to an edge with the maximum weight among all edges between~$S$ and $V(G)\setminus S$;
		\State $\pi(v) \leftarrow i$;			 {\hfill /* visit~$v$ */}
		\State $S \leftarrow S\cup \{v\}$;
		\EndFor
		\State \Return $\pi$;
	\end{algorithmic}
\end{algorithm}

\begin{lemma}[\cite{RONG22}]
\label{lem-prefix-prim-ordering}
	Let~$G$ be a connected chordal graph, and let~$(K_1, K_2, \dots, K_t)$ be a Prim ordering of~$C(G)$. 
	For every $i\in [t]$, the subgraph of~$C(G)$ induced by $\{K_1, K_2, \ldots, K_i\}$ is the clique graph of the subgraph of~$G$ induced by $\bigcup_{j\in [i]} K_j$.
\end{lemma}

Let~$\sigma$ be an ordering of~$V(G)$, and let~$\pi$ be an ordering of the maximal cliques of~$G$. For a vertex $v\in V(G)$, we use~$\fstcpri{\pi}{v}$ to denote the first $v$-clique in~$\pi$. 
We say that~$\sigma$ is a \emph{generation} of~$\pi$ (or~$\pi$ {\emph{generates}}~$\sigma$) if for all $x, y\in V(G)$ it holds that $\fstcpri{\pi}{x} <_\pi \fstcpri{\pi}{y}$ implies $x <_\sigma y$.   
Precisely, for an ordering $\pi = (K_1, K_2, \ldots, K_t)$ of maximal cliques of~$G$ and $i\in [t-1]$, let $\vsetprim{\pi}{i}=\bigcup_{j\in [i]}K_j$ be the set of vertices of~$G$ contained in at least one of the first~$i$ cliques from~$\pi$. In addition, for $i\in [t]\setminus \{1\}$, let $K(\pi, i)=K_i\setminus \vsetprim{\pi}{i-1}$ be the set of vertices of~$G$ contained in~$K_i$ but not in any other cliques before~$K_i$ in~$\pi$. Besides, let $K(\pi, 1)=K_1$. Then,~$\sigma$ is a generation of~$\pi$ if and only if it is of the form $(\overrightarrow{K(\pi, 1)}, \overrightarrow{K(\pi, 2)}, \dots, \overrightarrow{K(\pi, t)})$, where for a set~$X$,~$\overrightarrow{X}$ can be any ordering of~$X$.

\begin{lemma}[\cite{RONG22}]
\label{lem-mcs-prim-ordering}
Let~$G$ be a connected chordal graph. Then, an ordering of~$V(G)$ is an MCS ordering of~$G$ if and only if it is a generation of some Prim ordering of~$C(G)$.
\end{lemma}

The generic search is a search paradigm so that, except the first visited vertex which can be any arbitrary one, a vertex can be visited as long as at least one of its neighbors has already been visited~\cite{DBLP:journals/siamcomp/Tarjan72}.   
See Algorithm~\ref{alg-generic} for a pseudocode of the generic search.

\begin{algorithm}[htbp]
	\caption{Generic Search}
	\label{alg-generic}
	\prealg
	{A connected graph~$G$ of~$n$ vertices}
	{An ordering $\pi$ of~$V(G)$}
	\begin{algorithmic}[1]
	    \State let~$v$ be any arbitrary vertex in $G$;
	    \State $\pi(v) \leftarrow 1$; {\hfill  /* visit~$v$ */}
		\For{$i = 2, \dots, n$} 
		\State let~$v$ be any unvisited vertex having at least one visited neighbor;
		\State $\pi(v) \leftarrow i$; {\hfill  /* visit~$v$ */}
		\EndFor
		\State \Return $\pi$;
	\end{algorithmic}
\end{algorithm}

\section{Layer Structures of Clique Graphs}
\label{sec-ls}

In this section, we introduce the notion of layer structure, and explore a number of structural properties pertinent to this notion. 
Throughout this section, let~$G$ be a connected chordal graph, and let~$K^{\star}$ be a maximal clique of~$G$. 

Let~$\mathcal{U}$ be the set of all connected components of~$C(G)$ after the deletion of all edges with the minimum weight. For the sake of readability, let us call each $U\in \mathcal{U}$ a {\emph{unit}}. If a maximal clique~$K$ of~$G$ is a node in a unit, we say that~$K$ is {\it{contained}} in this unit. We use~$U^K$ to denote the unit containing~$K$. Besides, we use~$\mathcal{K}(U)$ to denote the set of maximal cliques of~$G$ contained in a unit~$U$, and use~$\vv(U)$ to denote the set of vertices of~$G$ contained in nodes of~$U$, i.e., $\vv(U)=\bigcup_{K\in \mathcal{K}(U)} K$. 
We say that an edge in~$C(G)$ {\emph{crosses}} two units if the two endpoints of the edge are respectively from the two units.

The layer structure we shall study is a refinement of the clique graph~$C(G)$. The following two lemmas pinpoint where the refinement lies. 

\begin{lemma}
\label{ppt-unit-larger-weight}
	The weights of edges of~$C(G)$ whose both endpoints are contained in the same unit are greater than the minimum weight of edges of~$C(G)$.
\end{lemma}

\begin{proof}
	Let~$\edge{K}{K'}$ be an edge with the minimum weight in~$C(G)$, and let $S = K \cap K'$. 
	Therefore,~$S$ is a $u$-$v$ separator for all $u \in K \setminus S$ and $v \in K' \setminus S$. 
	Then, by Lemma~\ref{lem-clique-graph-delete-edge} (ii),~$K$ and~$K'$ are in different connected components of $C(G) \ominus S$, implying that~$K$ and~$K'$ are contained in different units.
\end{proof}

Lemma~\ref{ppt-unit-larger-weight} equivalently asserts that every unit~$U$ is exactly the subgraph of~$C(G)$ induced by~$\mathcal{K}(U)$. Or, to put it another way, an edge in~$C(G)$  crosses two units if and only if it has the minimum weight in~$C(G)$. 

\begin{lemma}
\label{ppt-same-label-bet-unit}
Let~$U$ and~$U'$ be two units from~$\mathcal{U}$ so that there are edges in~$C(G)$ crossing~$U$ and~$U'$. Then, all edges in~$C(G)$ crossing~$U$ and~$U'$ have the same label.
\end{lemma}

\begin{proof}
	Towards a contradiction, assume that~$C(G)$ contains two distinct edges~$\edge{K_1}{K_2}$ and~$\edge{K_3}{K_4}$ crossing~$U$ and~$U'$ with different labels. 
	Let $S = K_1 \cap K_2$ and let $S' = K_3 \cap K_4$. So, $S \neq S'$. 
	Obviously,~$U$ and~$U'$ both remain connected in $C(G) \ominus S$.  
	Moreover, since the edge $\edge{K_3}{K_4}$ is present in $C(G) \ominus S$,~$U$ and~$U'$ are in the same connected component of $C(G) \ominus S$. 
 This means that~$K_1$ and~$K_2$ are in the same connected component of $C(G)\ominus S$. 
	However,  this contradicts Lemma~\ref{lem-clique-graph-delete-edge}~(ii). 
\end{proof}

Now we are ready to define the layer structure.

\begin{definition}[Layer Structure]
\label{def-layer-structure}
    The layer structure of~$C(G)$ rooted by~$K^{\star}$ is a graph with the vertex set~$\mathcal{U}$ so that there is an edge between two units in~$\mathcal{U}$ if and only if there exists at least one edge in~$C(G)$ crossing the two units. The label and the weight of an edge $\edge{U}{U'}$ in the layer structure are respectively $K\cap K'$ and $\abs{K\cap K'}$ where $\edge{K}{K'}$ can be any edge in~$C(G)$ crossing~$U$ and~$U'$. The unit~$U^{K^{\star}}$ is called the {{root}} of the layer structure. A unit is in the $i$-th layer if it is at a distance~$i$ from the root, where the distance between two units is defined as the length of a shortest path between them in the layer structure. 
\end{definition}

For an illustration of Definition~\ref{def-layer-structure}, see Figure~\ref{fig-layer-structure}. Note that due to Lemma~\ref{ppt-same-label-bet-unit}, the labels and the weights of edges in the layer structure are well-defined. 

Let~$\mathcal{L}_i$ be the set of all units in the $i$-th layer, and let $\mathcal{L}_{\leq i}=\bigcup_{j\in [i]\cup \{0\}}\mathcal{L}_j$. Obviously, $\mathcal{L}_0=\{U^{K^{\star}}\}$. In addition, if two units from respectively two layers~$\mathcal{L}_i$ and~$\mathcal{L}_j$ are adjacent, it holds that $|i-j|\leq 1$. 
Recall that~$G$ is connected. Hence,~$C(G)$ and the layer structure are connected too. As a result, every unit in~$\mathcal{U}$ is in some layer.    

\begin{figure}[htbp]
\centering{
\includegraphics[width=0.5\textwidth]{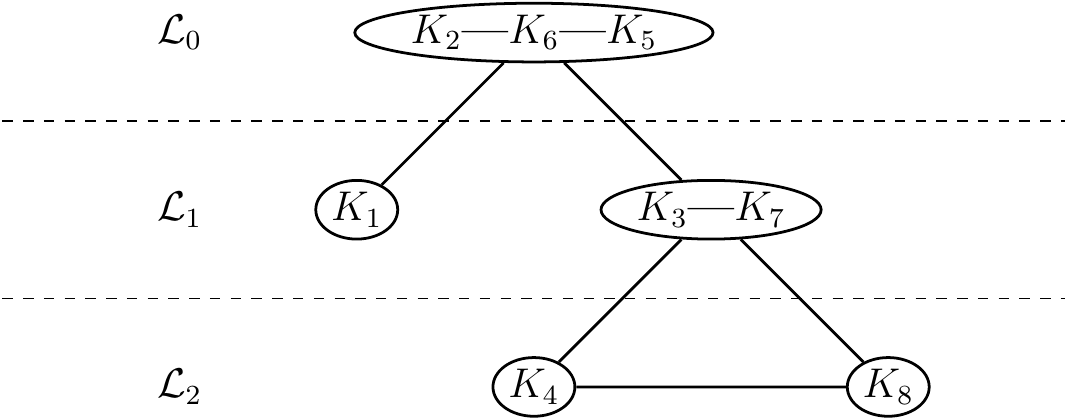}
}
	\caption{The layer structure of the chordal graph in Figure~\ref{fig-clique-graph-example} rooted by~$K_2$ (or~$K_5$,~$K_6$).}
	\label{fig-layer-structure}
\end{figure}

Lemmas~\ref{ppt-unit-larger-weight} and~\ref{ppt-same-label-bet-unit} suggest that the layer structure is nothing special but a reinspection of the clique graph~$C(G)$ by (1) grouping nodes into units which are connected components of~$C(G)$ without edges of the minimum weight; (2) fixing a unit as the root and arranging units into layers according to their distances to the root; and (3) regarding edges crossing two units as one edge between the units. 
That said, such a reinspection shed light on numerous significant properties which can be exploited for solving {\prob{PSOP-MCS}}. 
Now let us start the exploration on these properties.

\begin{property}
\label{cor-unit-separator} 
	Let $\edge{U}{U'}$ be an edge in the layer structure. Then, every path between~$U$ and~$U'$ in the layer structure contains an edge with the same label as~$\edge{U}{U'}$.
\end{property}

\begin{proof}
Let~$\edge{K}{K'}$ be an edge in~$C(G)$ crossing the two units~$U$ and~$U'$. 
    	Let $S = K \cap K'$. It is clear that~$S$ is both the label of~$\edge{K}{K'}$ and the label of~$\edge{U}{U'}$. 
	By Lemma~\ref{lem-clique-graph-delete-edge}~(ii),~$K$ and~$K'$ are disconnected in $C(G) \ominus S$. 
	It follows that every $K$-$K'$ path in~$C(G)$ contains at least one edge with the label~$S$. 
	Then, by Lemma~\ref{ppt-unit-larger-weight}, we know that every $U$-$U'$ path in the layer structure contains at least one edge with the label~$S$.
\end{proof}

\begin{property}
\label{ppt-layer-adjacency} 
	Let $\edge{U_1}{U_2}$ and $\edge{U_3}{U_4}$ be two distinct edges of the layer structure with the same label~$S$.
	Then, the units in the set $\{U_1, U_2, U_3, U_4\}$ are pairwise adjacent in the layer structure, and all edges among them have the same label~$S$. 
\end{property}

\begin{proof}
	Let $\edge{K_1}{K_2}$ and $\edge{K_3}{K_4}$ be two distinct edges of~$C(G)$ both with the label~$S$, where $K_i \in \mathcal{K}(U_i)$ for every $i \in [4]$. We first show that~$U_1$,~$U_2$,~$U_3$, and~$U_4$ are pairwise adjacent in the layer structure. Assume, for the sake of contradiction, that one from $\{U_1, U_2\}$ is not adjacent to one from $\{U_3, U_4\}$. By symmetry, suppose that $U_1\neq U_3$, and~$U_1$ is not adjacent to~$U_3$ in the layer structure. 
	
	We claim that~$K_1$ and~$K_3$ are disconnected in $C(G) \ominus S$. Assume for contradiction that~$K_1$ and~$K_3$ are connected in $C(G)\ominus S$. Let~$C'$ be the connected component of $C(G) \ominus S$ containing~$K_1$ and~$K_3$, and let~$V'$ be the set of vertices of~$G$ contained in nodes of~$C'$. By the minimality of~$|S|$, there exists a Prim ordering of~$C(G)$ so that all maximal cliques of~$G$ contained in~$C'$ are visited before all the other maximal cliques of~$G$. Additionally, in light of Lemma~\ref{lem-clique-graph-delete-edge}~(ii), the subgraph of~$C(G)$ induced by~$C'$ does not contain any edge with the label~$S$. 
    Then, by Lemma~\ref{lem-prefix-prim-ordering}, we know that~$C'$ is the clique graph of~$G[V']$. Obviously, $S \subsetneq V'$. By Lemma~\ref{lem-clique-graph},~$C'$ contains all clique trees of~$G[V']$ as subgraphs. 
    We fix a clique tree of~$G[V']$. There is a unique $K_1$-$K_3$ path in the clique tree, and by the definition of clique trees, all nodes on this path are $S$-cliques. Obviously, this path is also present in~$C(G)\ominus S$. 
    The label of each edge on this path cannot be~$S$, since such edges are absent in $C(G)\ominus S$. So, the labels of all edges on this path properly contain~$S$. This indicates that~$K_1$ and~$K_3$ are contained in the same unit in the layer structure. However, this contradicts that $U_1 \neq U_3$.

	So, we know that~$K_1$ and~$K_3$ are disconnected in $C(G) \ominus S$. Then, according to Lemma~\ref{lem-clique-graph-delete-edge} (ii),~$S$ is a~$u$-$v$ separator for all $u \in K_1 \setminus S$ and all $v \in K_3 \setminus S$. 
	By the minimality of~$|S|$ and the definition of clique graphs,~$\edge{K_1}{K_3}$ is an edge of~$C(G)$ with the label~$S$. As a result, $\edge{U_1}{U_3}$ is an edge of the layer structure, contradicting that~$U_1$ and~$U_3$ are nonadjacent in the layer structure. 
	
	Now we can conclude that the units in $\{U_1, U_2, U_3, U_4\}$ are pairwise adjacent in the layer structure. 
	Then, from Property~\ref{cor-unit-separator} and the fact that both~$\edge{U_1}{U_2}$ and~$\edge{U_3}{U_4}$ have the label~$S$, it follows that the labels of the edges among~$U_1$,~$U_2$,~$U_3$, and~$U_4$ are all~$S$.
\end{proof}

Note that in Property~\ref{ppt-layer-adjacency} it may be that $\{U_1, U_2\}\cap \{U_3, U_4\}\neq\emptyset$. 

\begin{property}
\label{ppt-layer-parent}
Every unit~$U$ in the $\ii$-th layer~$\mathcal{L}_{\ii}$ where ${\ii}\geq 1$ is adjacent to exactly one unit from the layer~$\mathcal{L}_{\ii-1}$. 
\end{property}

\begin{proof}
Let~$U$ be a unit from the $\ii$-th layer. 
	The statement trivially holds  for $\ii = 1$. It remains to consider the case where $\ii \geq 2$.
	Assume, for the sake of contradiction, that there are two distinct units~$U_1$ and~$U_2$ in the layer~$\mathcal{L}_{\ii-1}$ both adjacent to~$U$ in the layer structure. Let~$S_1$ and~$S_2$ be the labels of~$\edge{U}{U_1}$ and~$\edge{U}{U_2}$, respectively.
	By Definition~\ref{def-layer-structure}, there is a path from the root to~$U_1$, and a path from the root to~$U_2$ in the layer structure. Then, as~$\ii\geq 2$, there exists a $U_1$-$U_2$ path~$P$ of length at least two in the layer structure such that all inner units of~$P$ are from~$\mathcal{L}_{\leq \ii-2}$. Our proof is completed by distinguishing between the following two cases. 
	
	\begin{description}
		\item[Case 1:] $S_1 = S_2$.
		
		By Property~\ref{ppt-layer-adjacency}, $\edge{U_1}{U_2}$ is an edge with the label~$S_1$ in the layer structure. By Property~\ref{cor-unit-separator}, there exists an edge in~$P$ with the label~$S_1$.
		By Property~\ref{ppt-layer-adjacency},~$U$ is adjacent to the two endpoints of this edge, which is impossible since~$U$ is from the $\ii$-th layer but at least one of the two endpoints of the edge is from~$\mathcal{L}_{\leq \ii-2}$.
		
		\item[Case 2:] $S_1 \neq S_2$.
		
		By Property~\ref{cor-unit-separator}, there exists an edge in $E(P) \cup \{\edge{U}{U_2}\}$ with the label~$S_1$. As $S_1 \neq S_2$, this edge belongs to~$P$.
		By Property~\ref{ppt-layer-adjacency},~$U$ is adjacent to the two endpoints of this edge. However, analogous to the discussion in Case~1, we know that this is impossible. 
	\end{description}
As both cases lead to some contradictions, we know that~$U$ is adjacent to exactly one unit from~$\mathcal{L}_{\ii-1}$. 
\end{proof}

Now for each unit~$U$ from a layer $\mathcal{L}_{\ii}$ where ${\ii}\geq 1$, we call the only unit from the layer $\mathcal{L}_{{\ii}-1}$ adjacent to~$U$ the {\emph{parent}} of~$U$, and use~$\parent{U}$ to denote it. Correspondingly, we say that~$U$ is a {\emph{child}} of~$\parent{U}$. Furthermore, for a unit~$U$ from a layer~$\mathcal{L}_i$ and a unit~$U'$ from a layer~$\mathcal{L}_j$ such that $i<j$, we say that~$U'$ is a {\emph{descendant}} of~$U$ if there is a path from the root to~$U'$ through~$U$ (i.e.,~$U$ is also on the path) in the layer structure. For a unit~$U$, let~$\dst{U}$ be the set of all descendants of~$U$, and let $\cdst{U}=\dst{U}\cup \{U\}$.
Note that every unit is in a higher layer (assuming the root is in the highest layer) than any of its descendants.

\begin{property}
\label{ppt-bag}
	Let~$U_1$ and~$U_2$ be two units in the same layer~$\mathcal{L}_i$ where $i\geq 1$. Then, the following statements are equivalent:
	\begin{enumerate}
	    \item[(1)] $U_1$ and $U_2$ are adjacent.
	    \item[(2)] $\parent{U_1}=\parent{U_2}$, and the label of the edge between~$U_1$ and its parent equals that between~$U_2$ and its parent. 
	\end{enumerate}
\end{property}

\begin{proof}
From Property~\ref{ppt-layer-adjacency}, we know that Statement~(2) implies Statement~(1). In the following, we show that Statement~(1) implies Statement~(2).     

    Assume that~$U_1$ and~$U_2$ are adjacent, and let~$S$ be the label of the edge between them. Let $U_3 = {\parent{U_1}}$ and let $U_4 = {\parent{U_2}}$.  
		We first prove that~$U_1$ and~$U_2$ have the same parent, i.e., $U_3=U_4$. Assume, for contradiction, that $U_3 \neq U_4$. By Properties~\ref{ppt-layer-adjacency} and~\ref{ppt-layer-parent}, this implies that neither the label of~$\edge{U_1}{U_3}$ nor the label of~$\edge{U_2}{U_4}$ is~$S$. In addition, it also holds that $i > 1$, implying that there exists a $U_3$-$U_4$ path~$P$ of length at least two in the layer structure whose inner units are all from~$\mathcal{L}_{\leq i-2}$ (cf.\ the proof of Property~\ref{ppt-layer-parent}). Then, by Property~\ref{cor-unit-separator}, there exists an edge in~$P$ with the label~$S$. By Property~\ref{ppt-layer-adjacency}, the two endpoints of this edge are adjacent to~$U_1$ in the layer structure. However, this is impossible since at least one endpoint of this edge is from $\mathcal{L}_{\leq i-2}$ but~$U_1$ is from the layer~$\mathcal{L}_i$. This completes the proof for that $U_3=U_4$. 
	    Having $U_3=U_4$, from Properties~\ref{cor-unit-separator} and~\ref{ppt-layer-adjacency}, it follows that the labels of~$\edge{U_1}{U_3}$ and~$\edge{U_2}{U_4}$ are both~$S$. 
\end{proof}

Property~\ref{ppt-bag} implies that if a subset of units in the same layer~$\mathcal{L}_i$ are connected in the layer structure restricted to~$\mathcal{L}_i$, then they are pairwise adjacent. 
For ease of exposition, we group units in the same layer into bags so that two units are in the same bag if they are adjacent. By~Property~\ref{ppt-bag}, all units in the same bag have the same parent. By Property~\ref{ppt-layer-parent}, if we ignore edges inside all bags in the layer structure, we obtain a tree rooted at~$U^{K^{\star}}$. An important consequence of this fact is that every path connecting two units from the same layer~$\mathcal{L}_i$ is completely contained in~$\mathcal{L}_{\leq i}$. 

In addition, Property~\ref{ppt-layer-parent} indicates that every nonroot unit~$U$ in a layer~$\mathcal{L}_i$ is adjacent to exactly one unit in~$\mathcal{L}_{\leq i-1}$, and this unit is its parent in~$\mathcal{L}_{i-1}$. Property~\ref{ppt-bag} further strengthens that~$\parent{U}$ indeed separates all descendants of~$\parent{U}$ from all the other units, if there are any. In other words,~$\{\parent{U}\}$ is a $U'$-$U''$ separator in the layer structure for all $U'\in \dst{\parent{U}}$ and all $U''\in \mathcal{U}\setminus \cdst{\parent{U}}$ whenever $\mathcal{U}\setminus \cdst{\parent{U}}\neq\emptyset$. An important consequence of this fact is that every generic search ordering of the layer structure starting from the root visits every nonroot unit after (not necessarily consecutive) its parent.

The next property identifies the label of the edge between two adjacent units. 

\begin{property}
\label{ppt-label-between-units}
The label of every edge between two units~$U$ and~$U'$ in the layer structure is $\vv(U)\cap \vv(U')$.
\end{property}

\begin{proof}
    Let~$U$ and~$U'$ be two adjacent units in the layer structure. Let $S=\vv(U)\cap \vv(U')$, and let~$S'$ be the label of the edge between~$U$ and~$U'$ in the layer structure. Clearly,~$S'\subseteq S$. Therefore, to complete the proof, it needs only to show that every~$v\in S$ is contained in~$S'$. In light of  Properties~\ref{ppt-layer-parent} and \ref{ppt-bag}, only the two cases described below may occur. 
    \begin{description}
        \item[Case~1:] one of $U$ and $U'$ is the parent of the other.
        
        For the sake of contradiction, assume that there exists $v\in S\setminus S'$. By Lemma~\ref{ppt-same-label-bet-unit}, all edges between~$\mathcal{K}(U)$ and~$\mathcal{K}(U')$ in~$C(G)$ have the same label~$S'$. By the definition of the layer structure (Definition~\ref{def-layer-structure}),~$\abs{S'}$ is a minimum edge weight in~$C(G)$. Then, by Properties~\ref{ppt-layer-parent} and \ref{ppt-bag},~$U$ and~$U'$ are disconnected in the layer structure after removing all edges with the label~$S'$. This is equivalent to every $K\in \mathcal{K}(U)$ and every $K'\in \mathcal{K}(U')$ being disconnected in $C(G)\ominus S'$. However, as $v\in \vv(U)\cap \vv(U')$, this violates Lemma~\ref{lem-clique-graph-delete-edge}~(i).
        
         \item[Case~2:] $U$ and $U'$ are in the same bag. 
         
         By Property~\ref{ppt-bag},~$U$ and~$U'$ have the same parent, say~$\hat{U}$, and the labels of the edges among the units in the same bag as~$U$ and their parent~$\hat{U}$ are~$S'$. Then, similar to Case~1, it is easy to see that~$U$ and~$U'$ are disconnected in the layer structure after removing all edges with the label~$S'$, which violates Lemma~\ref{lem-clique-graph-delete-edge}~(i) too. 
    \end{description}
As both cases violate Lemma~\ref{lem-clique-graph-delete-edge}, we know that Property~\ref{ppt-label-between-units} holds.
\end{proof}

\begin{property}
\label{ppt-unit-contain-vertex} 
	Let~$v$ be a vertex in~$G$. Let~$\ii$ be the minimum integer such that~$\mathcal{L}_{\ii}$ contains a maximal $v$-clique. Then, all maximal $v$-cliques of~$G$ contained in~$\mathcal{L}_{\ii}$ are within one unit.
\end{property}

\begin{proof}
Let~$v$ and~$\ii$ be as stipulated in the statement of Property~\ref{ppt-unit-contain-vertex}.  
Assume for contradiction that there exist two maximal $v$-cliques~$K$ and~$K'$ respectively from two different units~$U$ and~$U'$ in the layer~$\mathcal{L}_{\ii}$. Our proof is completed by considering the following two cases. 
\begin{description}
    \item[Case~1:] $U$ and~$U'$ are adjacent in the layer structure.
   
   By Property~\ref{ppt-label-between-units},~$v$ is in the label of the edge $\edge{U}{U'}$. By Property~\ref{ppt-bag},~$U$ and~$U'$ have the same parent, say,~$\hat{U}$. Moreover, by Properties~\ref{ppt-layer-adjacency} and~\ref{ppt-bag}, the edges~$\edge{U}{\hat{U}}$ and $\edge{U'}{\hat{U}}$ have the same label as the edge~$\edge{U}{U'}$. It follows that~$\hat{U}$ contains some $v$-clique of~$G$. However, this contradicts that~$\ii$ is the minimum integer with~$\mathcal{L}_{\ii}$ containing a maximal $v$-clique of~$G$.
   
   \item[Case~2:] $U$ and~$U'$ are not adjacent in the layer structure.
   
   Note that in this case~$\ii\geq 1$.  As the subgraph of~$C(G)$ induced by the set of all maximal $v$-cliques is connected, there is a $K$-$K'$ path~$P$ in~$C(G)$ so that~$v$ is in the label of every edge on the path. By Properties~\ref{ppt-layer-parent} and~\ref{ppt-bag}, this path contains at least one maximal clique of~$G$ contained in~$\parent{U}$. 
However, similar to Case~1, this is in contradiction with the definition of~$\ii$.
\end{description}
As the above two cases cover all possibilities, and each of them leads to some contradictions, Property~\ref{ppt-unit-contain-vertex} holds. 
\end{proof}

Property~\ref{ppt-unit-contain-vertex} shows that for every vertex~$v$ in~$G$, there is a unique unit that contains a maximal $v$-clique and is at the least distance to the root in the layer structure. 
Let~$U^v$ be such a unique unit for~$v$.  

\begin{property}
    \label{pro-ancestor}
    For every vertex $v\in V(G)$, every unit containing a maximal $v$-clique of~$G$ is from~$\cdst{U^v}$.
\end{property}

\begin{proof}
    Assume, for the sake of contradiction, there exists $v\in V(G)$ and a unit~$U$ such that $v\in \vv(U)$ and~$U\not\in \cdst{U^v}$. By Property~\ref{ppt-unit-contain-vertex},~$U$ and~$U^v$ cannot be in the same layer. Then, as $U\not\in \cdst{U^v}$, every~$U^v$-$U$ path in the layer structure contains the parent of~$U^v$. As all maximal $v$-cliques are connected in~$C(G)$, and $v\in \vv(U)\cap \vv(U^v)$, we know that $v\in \vv(\parent{U^v})$. However, this contradicts the definition of~$U^v$.
 \end{proof}

\begin{property}
\label{ppt-unit-entrance}
	Let~$U$ and~$U'$ be two  units from the same bag, and let~$\hat{U}$ be their parent. Then, for every edge~$\edge{K}{K'}$ in~$C(G)$ so that $K \in \mathcal{K}(U)$ and $K' \in 
\mathcal{K}({U'})$, there exists $\hat{K}\in \mathcal{K}({\hat{U}})$ which is adjacent to both~${K}$ and~$K'$ in~$C(G)$.
\end{property}

\begin{proof}
	Let $K\in \mathcal{K}(U)$ and $K'\in \mathcal{K}(U')$ be as stipulated in Property~\ref{ppt-unit-entrance}. Let $S = K \cap K'$.
	As~$U$ and~$U'$ are in the same bag, they are adjacent in the layer structure. By Lemma~\ref{ppt-same-label-bet-unit} and Definition~\ref{def-layer-structure}, the label of the edge~$\edge{U}{U'}$ is~$S$.  This means that~$\abs{S}$ is the minimum weight of edges in~$C(G)$. 
	By Properties~\ref{ppt-layer-adjacency} and~\ref{ppt-bag}, the labels of the edges~$\edge{U}{\hat{U}}$ and~$\edge{U'}{\hat{U}}$ are also~$S$. 
	As a result, there exists~$\hat{K}\in \mathcal{K}({\hat{U}})$ so that $S \subseteq \hat{K}$. By Properties~\ref{ppt-layer-parent} and~\ref{ppt-bag},~$\hat{K}$ is disconnected from~$K$ in~$C(G) \ominus S$. Then, according to Lemma~\ref{lem-clique-graph-delete-edge} (ii) and the minimality of $|S|$,~$S$ is a minimal~$u$-$v$ separator in~$G$ for all $u \in K \setminus \hat{K}$ and all $v \in \hat{K} \setminus K$. Therefore,~$\hat{K}$ and~$K$ are adjacent in~$C(G)$. Similarly,~$\hat{K}$ is also adjacent to~$K'$ in~$C(G)$.
\end{proof}

Now we study two lemmas which provide insight into connections among MCS orderings of~$G$, Prim orderings of~$C(G)$, and generic search orderings of the layer structure.  
We say that a Prim ordering~$\pi$ of the clique graph~$C(G)$ {\emph{respects}} a partial order~$R$ {\oover}~$V(G)$ if for every $(x, y)\in R$ it holds that $\fstcpri{\pi}{x}<_{\pi} \fstcpri{\pi}{y}$ or $\fstcpri{\pi}{x} = \fstcpri{\pi}{y}$.  
By saying that an ordering starts with elements in a subset, we mean the elements in the subset are before all the other elements in the ordering. 

\begin{lemma}
\label{lem-mcs-prim-ordering-b}
	Let~$R$ be a partial order {\oover}~$V(G)$. 
	There is an MCS ordering of~$G$ extending~$R$ if and only if there is a Prim ordering of~$C(G)$ respecting~$R$.  
	Moreover, given a Prim ordering of~$C(G)$ which starts with some node~$K$ and respects~$R$, we can construct an MCS ordering of~$G$ starting with the vertices in~$K$ and extending~$R$ in polynomial time.
\end{lemma}

\begin{proof}
We start the proof with first statement. 
For the forward direction, let~$\sigma$ be an MCS ordering of~$G$ extending~$R$. By Lemma~\ref{lem-mcs-prim-ordering}, there exists a Prim ordering~$\pi$ of~$C(G)$ that generates~$\sigma$.  
	For every $(x, y)\in R$, since~$\sigma$ extends $R$, it holds that $x <_\sigma y$. As~$\pi$ generates~$\sigma$, it holds that $\fstcpri{\pi}{x} <_\pi \fstcpri{\pi}{y}$ or $\fstcpri{\pi}{x} = \fstcpri{\pi}{y}$. 
	For the backward direction, let $\pi = (K_1, K_2, \ldots, K_t)$ be a Prim ordering of~$C(G)$ so that for all $(x,y)\in R$ it holds that either $\fstcpri{\pi}{x} <_\pi \fstcpri{\pi}{y}$ or 
	$\fstcpri{\pi}{x} = \fstcpri{\pi}{y}$. Let~$\sigma$ be a generation of~$\pi$ so that for every $i\in [t]$ it holds that~$\sigma$ restricted to~$K(\pi, i)$ extends~$R$ restricted to~$K(\pi, i)$, i.e.,  $\rsto{\sigma}{K(\pi, i)}$ is a linear extension of $\rsto{R}{K(\pi, i)}$. As for every $i, j\in [t]$ such that $i\neq j$, $K(\pi, i)$ and $K(\pi, j)$ are disjoint,~$\sigma$ is well-defined. 
	By Lemma~\ref{lem-mcs-prim-ordering},~$\sigma$ is an MCS ordering of~$G$. To complete the proof for the first statement, it suffices to show that~$\sigma$ is a linear extension of $R$. Let $(x, y)\in R$. If $\fstcpri{\pi}{x}=\fstcpri{\pi}{y}$, i.e.,~$x$ and~$y$ are contained in some $K(\pi, i)$ where $i\in [t]$, as $\rsto{\sigma}{K(\pi, i)}$ extends~$\rsto{R}{K(\pi, i)}$, it holds that $x <_{\sigma} y$. Otherwise, $\fstcpri{\pi}{x} <_{\pi} \fstcpri{\pi}{y}$ holds. Then, as~$\sigma$ is a generation of~$\pi$, $x <_{\sigma} y$ holds too. This completes the proof for that~$\sigma$ extends~$R$. Note that as~$\sigma$ is a generation of~$\pi$, all vertices in~$K_1$ are before all the other vertices of~$G$ in the ordering~$\sigma$.
	
	Concerning the second statement, observe that the above proof for the backward direction is constructive, and the polynomial-time solvability follows from the obvious fact that computing a linear extension of a partial order can be done in polynomial time.
\end{proof}

By Property~\ref{ppt-unit-contain-vertex}, 
each partial order~$R$ {\oover}~$V(G)$ specifies a partial order {\oover} the units: 
 \[\mathcal{Q}^R=\{(U^x, U^y) \mid  (x, y) \in R\}.\]

\begin{lemma} 
\label{lem-generic-search-on-layers-a} 
Let~$R$ be a partial order {\oover}~$V(G)$. Then, if there is an MCS ordering of~$G$ that starts with the vertices from~$K^{\star}$ and extends~${R}$, there is a generic search ordering of the layer structure rooted by~$K^{\star}$ that starts with~$U^{K^{\star}}$ and extends~$\mathcal{Q}^R$. 
\end{lemma}

\begin{proof}
Let~$\sigma$ be an MCS ordering of~$G$ that starts with the vertices from~$K^{\star}$ and extends~${R}$. By Lemma~\ref{lem-mcs-prim-ordering}, there is a Prim ordering~$\pi$ of~$C(G)$ that generates~$\sigma$. Obviously,~$K^{\star}$ is the first node in~$\pi$. Then, by Algorithm~\ref{alg-Prim-Search}, Lemma~\ref{ppt-unit-larger-weight}, and Definition~\ref{def-layer-structure}, we know that for every unit~$U$, the nodes in~$U$ are consecutive in~$\pi$. Moreover, in view of Properties~\ref{ppt-layer-parent} and~\ref{ppt-bag}, for every nonroot unit~$U$, all nodes from~$\parent{U}$ are before all nodes from~$U$ in~$\pi$. Consequently, there is a linear order $\pi'=(U_1, U_2, \dots, U_t)$  of the units in the layer structure so that 
\begin{enumerate}
    \item[(1)] $U_1=U^{K^{\star}}$;
    \item[(2)] for every nonroot unit~$U$,~$\parent{U}$ is before~$U$ in~$\pi'$; and 
    \item[(3)] for every $i\in [t]$, all nodes of~$U_i$ are consecutive in the Prim ordering~$\pi$.  
\end{enumerate}
Conditions~(1) and~(2) mean that~$\pi'$ is a generic search ordering of the layer structure starting with the root. 
It remains to show that~$\pi'$ extends~$\mathcal{Q}^R$. For this purpose, 
let $(x, y)\in R$ such that $U^x\neq U^y$. As~$\sigma$ extends~$R$,~$x$ is before~$y$ in~$\sigma$. As~$\pi$ generates~$\sigma$, either $K_{\pi}^x=K_{\pi}^y$ holds or $K_{\pi}^x <_{\pi} K_{\pi}^y$ holds. Condition~(2) given above and Property~\ref{pro-ancestor} imply that~$K_{\pi}^v$ is contained in~$U^v$ for all~$v\in V(G)$. Then, as $U^x\neq U^y$, it must be that $K_{\pi}^x <_{\pi} K_{\pi}^y$. Finally, by Condition~(3) given above, we know that~$U^x$ is before~$U^y$ in~$\pi'$. This completes the proof that~$\pi'$ extends~$\mathcal{Q}^R$. 
\end{proof}

\section{A Dynamic Programming Algorithm for the {\prob{PSOP}}}

In this section, we present a polynomial-time dynamic programming algorithm for {\prob{PSOP-MCS}} restricted to chordal graphs. 

For an ordering~$\pi$ of units in a layer structure, and an ordering~$\sigma$ of nodes contained in a unit~$U$, {\emph{realizing}}~$\pi$ by~$\sigma$ is the operation of replacing~$U$ in~$\pi$ with~$\sigma$. For instance, for $\pi=(U_1, U_2, U_3, U_4)$ and $\sigma=(K_1, K_2, K_3)$ where~$\{K_1, K_2, K_3\}$ is the set of nodes in~$U_2$, realizing~$\pi$ by~$\sigma$ results in the linear order $(U_1, K_1, K_2, K_3, U_3, U_4)$.

\begin{theorem}
\label{thm-main}
{\prob{PSOP-MCS}} restricted to chordal graphs is polynomial-time solvable. 
\end{theorem}

\begin{proof}
Let $I=(G, R)$ be an instance of {\prob{PSOP-MCS}}, where~$G$ is a connected chordal graph of~$n$ vertices, and~$R$ is a partial order {\oover}~$V(G)$. 
Our algorithm consists of the following steps. 
\begin{description}
    \item[Step~1] We sort the weights of edges in the clique graph~$C(G)$ of~$G$ in increasing order. Let $(w(1)$, $w(2)$, $\dots$, $w(t))$ be this order, where~$t$ is the number of different weights of edges in~$C(G)$. Hence  $w(1)<w(2)<\cdots <w(t)$ holds. Notice that $t=\bigo{n}$ since edges of~$C(G)$ may have at most~$n-2$ different weights.  
    
    \item[Step~2] For each $i\in [t]$, let $C^{\mneg{i}}(G)$ be the graph obtained from~$C(G)$ by removing edges whose weights are from $\{w(1), w(2), \dots, w(i)\}$. 
    Let $C^{\mneg{0}}(G)=C(G)$. Clearly, every $C^{\mneg{i}}(G)$, where $i\in [t]$, is obtained from~$C^{\mneg{i-1}}(G)$ by deleting all edges with the minimum weight.
    
    \item[Step~3] We maintain a binary dynamic programming table $D(i, H, K)$, where $i\in [t]\cup \{0\}$,~$H$ is a connected component of~$C^{\mneg{i}}(G)$, and~$K$ is a node in~$H$. (Notice that for $i=0$, we have that $H=C(G)$.) As $t=\bigo{n}$, every chordal graph of~$n$ vertices has at most~$n$ maximal cliques, and~$K$ is a node in~$H$, the table has~$\bigo{n^2}$ entries. 

    For each connected component~$H$ of some~$C^{\mneg{i}}(G)$, let~$\vv(H)$ be the subset of vertices of~$G$ contained in nodes of~$H$. 
    We define $D(i, H, K)$ to be~$1$ if there is a Prim ordering of~$H$ which starts with~$K$ and respects~$\rsto{R}{\vv(H)}$, and define $D(i, H, K)$ to be~$0$ otherwise. (We elaborate on how to compute the table later.)
    
    \item[Step~4] After the table is computed, if $D(0, C(G), K)=0$ for all maximal cliques~$K$ of~$G$, by the definition of the table there is no Prim ordering of~$C(G)$ respecting~$R$, and by Lemma~\ref{lem-mcs-prim-ordering-b}, the given instance~$I$ is a {\noins}. Otherwise, there exists a maximal clique~$K$ of~$G$ such that $D(0, C(G), K)=1$. By the definition of the table and by Lemma~\ref{lem-mcs-prim-ordering-b}, there is an MCS ordering of~$G$ extending~$R$. Therefore, in this case, we conclude that~$I$ is a {\yesins}.  
\end{description}

Computing the table defined above lies at the core of the algorithm, which is the main focus of the remainder of the proof. To this end, we first show that each graph~$H$ used in Step~3 is a clique graph of a connected chordal graph. 

\begin{claim}
\label{claim-a}
For every $i\in [t]\cup \{0\}$, every connected component~$H$ of $C^{\mneg{i}}(G)$ is the clique graph of the subgraph of~$G$ induced by~$\vv(H)$. 
\end{claim}

\smallskip
\noindent{\it{Proof of Claim~\ref{claim-a}}}
We prove the claim by induction on~$i$. The statement is clearly true for $i=0$.  Now let $i\in [t]$, and let~$H$ be a connected component of~$C^{\mneg{i}}(G)$. Let~$H'$ be the connected component of~$C^{\mneg{i-1}}(G)$ containing~$H$. By induction,~$H'$ is the clique graph of~$G[\vv(H')]$ which is a connected chordal graph. 
Obviously,~$H$ is a connected component of~$H'$ after removing all edges with the minimum weight in~$H'$. As a consequence, there is a Prim ordering of~$H'$ which starts with a node from~$H$ and consecutively visits the nodes in~$H$. Then, by Lemma~\ref{lem-prefix-prim-ordering},~$H$ is the clique graph of the subgraph of~$G[\vv(H')]$ induced by~$\vv(H)$ which is exactly~$G[\vv(H)]$. This completes the proof for the claim.
\smallskip

By~Claim~\ref{claim-a}, each~$H$ in Step~3 is the clique graph of a connected chordal graph. Therefore, all properties and lemmas studied in the previous section apply to~$H$ and each of its layer structures. 

Now we are ready to show how to compute the table defined in Step~3. We fill the entries $D(i, H, K)$ of the table in a decreasing order of the values of~$i$. Entries where the first component is~$t$ are base entries, which are computed directly by the definition of the table. 

\begin{itemize}
    \item filling the base entries

Obviously, $C^{\mneg{t}}(G)$ consists of pairwise nonadjacent nodes corresponding to maximal cliques of~$G$. 
By the definition of the table, we directly set $D(t, H, K)=1$ for all base entries. 

     \item updating the table
     
Now we delineate how to update other entries $D(i, H, K)$, assuming all entries $D(i+1, H', K')$ have been computed.  
To compute the entry~$D(i, H, K)$, we first compute the layer structure of~$H$ rooted by~$K$. Let~$\mathcal{LS}$ be the layer structure. By Definition~\ref{def-layer-structure},~$\mathcal{LS}$ can be computed in polynomial time. Note that each unit in~$\mathcal{LS}$ is a connected component of~$C^{\mneg{i+1}}(G)$. For each vertex~$x$ in~$\vv(H)$, let~$U^x$ be the unit in~$\mathcal{LS}$ which contains~$x$ and is at the least distance from the root of~$\mathcal{LS}$. Recall that by Property~\ref{ppt-unit-contain-vertex}, such a unit is unique. Besides, for each node~$K'$ of~$H$, let~$U^{K'}$ be the unit in~$\mathcal{LS}$ containing~$K'$. Let $R'=\rsto{R}{\vv(H)}$ be~$R$ restricted to~$\vv(H)$. 
Let $\mathcal{Q}^{R'} = \{(U^x, U^y) \mid  (x,y) \in R'\}$.  
Now we determine if there is a generic search ordering of~$\mathcal{LS}$ which starts with the root and extends~$\mathcal{Q}^{R'}$. This can be done in polynomial time~\cite[Theorem 6]{DBLP:conf/wg/Scheffler22}\footnote{Scheffler~\cite{DBLP:conf/wg/Scheffler22} showed that the rooted version of the {\prob{PSOP}} for the generic search can be solved in polynomial time. In this version, we are given a graph~$G$, a partial order {\oover}~$V(G)$, and a vertex~$v\in V(G)$, and the question is whether~$G$ admits a generic search ordering which starts with~$v$ and extends the given partial order.}. If this is not the case, by Claim~\ref{claim-a} and Lemma~\ref{lem-generic-search-on-layers-a}, there is no MCS ordering of~$G[\vv(H)]$ that starts with the vertices from~$K$ and extends~$R'$, and by Claim~\ref{claim-a} and Lemma~\ref{lem-mcs-prim-ordering-b} there is no Prim ordering of~$H$ which starts with~$K$ and respects~$R'$. So in this case, we set $D(i, H, K)=0$. 
Otherwise, let $\pi=(U_0, U_1, \dots, U_{\pp})$ be a generic search ordering of~$\mathcal{LS}$ extending~$\mathcal{Q}^{R'}$ so that $U_0=U^K$.   	
Then, we let $D(i, H, K)=1$ if and only if 
\begin{enumerate}
\item[(1)] $D(i+1, U^K, K)=1$, and 
\item[(2)] for every~$U_j$, $j\in [\pp]$, there exists a node~$K_j$ contained in~$U_j$ such that in~$H$ the node~$K_j$ is adjacent to at least one node from the parent of~$U_j$ in~$\mathcal{LS}$ and, moreover, $D(i+1, U_j, K_j)=1$. 
\end{enumerate}
We show the correctness of this step as follows. 
Observe that in every generic search ordering of~$\mathcal{LS}$ starting from the root~$U^K$, every unit is visited before all its children (if there are any). 

We first prove the ``if'' direction. Assuming Conditions~(1) and~(2), let~$\pi'$ be the ordering obtained from~$\pi$ by 
\begin{enumerate}
    \item[(a)] realizing the first unit~$U^K$ by a Prim ordering of~$U^K$ which starts with~$K$ and respects~$R'$ restricted to~$U^K$ (guaranteed by $D(i+1, U^K, K)=1$), and 
    \item[(b)] realizing every~$U_j$, where $j\in [\pp]$, by a Prim ordering of~$U_j$ which starts with~$K_j$ and respects~$R'$ restricted to~$U_j$ (guaranteed by $D(i+1, U_j, K_j)=1$). 
\end{enumerate} 

The remainder of the proof for the ``if'' direction comprises the following two claims. 

\begin{claim}
\label{claim-b}
    $\pi'$ is a Prim ordering of~$H$ with the first node being~$K$.
\end{claim}

{\it{Proof of Claim~\ref{claim-b}}} 
By Condition~(1) and Operation~(a), we know that the first node in~$\pi'$ is~$K$. 
Besides, from Condition~(2) and Operation~(b), for every~$U_j$ where $j\in [\pp]$ the first node of~$\pi'$ restricted to~$U_j$ is~$K_j$.  By the definition of Prim ordering (Algorithm~\ref{alg-Prim-Search}) and the definition of~$\pi'$, it suffices now to show that for every~$U_j$ where $j\in [\pp]$ , the following condition holds:~$K_j$ is adjacent to at least one node in~$H$ which is before~$K_j$ in~$\pi'$ and is from a different unit adjacent to~$U_j$. This is the case as by Condition~(2),~$K_j$ is adjacent to at least one node from the parent of~$U_j$ in~$\mathcal{LS}$, and as~$\pi$ is a generic search ordering of~$\mathcal{LS}$ with the root being the first unit, by Properties~\ref{ppt-layer-parent} and~\ref{ppt-bag}, the parent of~$U_j$ is before~$U_j$ in~$\pi$, implying that all nodes in the parent of~$  U_j$ are before all nodes of~$U_j$ in~$\pi'$.

\begin{claim}
\label{cliam-c}
    $\pi'$ respects~$R'$.
\end{claim}

{\it{Proof of Claim~\ref{cliam-c}}} 
To verify that~$\pi'$ respects~$R'$, let $(x,y)\in R'$. Due to Properties~\ref{ppt-layer-parent}, \ref{ppt-bag}, and~\ref{pro-ancestor}, and that~$\pi$ is a generic search ordering of the layer structure~$\mathcal{LS}$ beginning with the root, we know that for every~$v\in V(G)$ the first node in~$\pi'$ containing~$v$ is from~$U^v$, i.e., $K^v_{\pi'}\in \mathcal{K}(U^v)$. Our proof proceeds by distinguishing between the following two cases. If~$U^x=U^y=U$, then as~$\pi$ has been realized by a Prim ordering of~$U$ respecting~$R'$ restricted to~$U$ in~$\pi'$, it holds that $K^x_{\pi'} <_{\pi'} K^y_{\pi'}$ or $K^x_{\pi'} = K^y_{\pi'}$. Otherwise, as~$\pi$ extends~$\mathcal{Q}^{R'}$, we know that~$U^x$ is before~$U^y$ in~$\pi$. By the definition of~$\pi'$, maximal~$x$-cliques in~$U^x$ are before maximal $y$-cliques in~$U^y$. By Properties~\ref{ppt-layer-parent},~\ref{ppt-bag}, and~\ref{pro-ancestor}, none of any units containing a maximal $y$-clique is before~$U^y$ in~$\pi$.  Then, from $K^x_{\pi'}\in \mathcal{K}(U^x)$ and $K^y_{\pi'}\in \mathcal{K}(U^y)$, it follows that $K^x_{\pi'} <_{\pi'} K^y_{\pi'}$. 
\end{itemize}

Now we give the proof for the ``only if'' direction. To this end, assume that $D(i, H, K)=1$, i.e.,~$H$ admits at least one Prim ordering, say~$\pi'$, which starts with~$K$ and respects~$R'$. As~$\pi'$ respects~$R'$, for each unit~$U$ in the layer structure~$\mathcal{LS}$,~$\pi'$ restricted to~$U$, i.e.,~$\pi'|_{U}$, is a Prim ordering of~$U$ respecting~$R'$ restricted to~$U$. Consequently, $D(i+1, U, K')=1$ where~$K'$ is the first node in~$\pi'|_{U}$. This immediately implies that Condition~(1) holds. We show below that Condition~(2) also holds. Let~$U_j$, $j\in [\pp]$, be a unit in~$\mathcal{LS}$. Let~$K'$ be the first node in~$\pi'|_{U_j}$. We claim that~$K'$ is adjacent in~$H$ to some node from the parent of~$U_j$ in~$\mathcal{LS}$. First, as~$\pi'$ is a Prim ordering of~$H$ and~$K'$ is not the first node in~$\pi'$,~$K'$ is adjacent to at least one node, say~$\hat{K}$, before~$K'$ in~$\pi'$ and, moreover, as~$K'$ is the first node in~$\pi'|_{U_j}$,~$\hat{K}$ is from a different unit, say~$\hat{U}$. If~$\hat{U}$ is the parent of~$U_j$ in~$\mathcal{LS}$, we are done.  
Otherwise, by Lemma~\ref{ppt-unit-larger-weight} and Definition~\ref{def-layer-structure}, we know that nodes in each unit are consecutive in~$\pi'$. 
As~$K$ is the first node of~$\pi'$ and~$K$ is contained in the root of~$\mathcal{LS}$, by Properties~\ref{ppt-layer-parent} and~\ref{ppt-bag}, none of the nodes contained in any descendant of~$U_j$ is visited before~$K'$ in~$\pi'$. It follows that~$\hat{U}$ is in the same bag as~$U_j$. Then, by Property~\ref{ppt-unit-entrance}, there exists a node from the parent of~$U_j$ which is adjacent to~$K'$ in~$H$. 

The algorithm runs in polynomial time since the table has at most~$\bigo{n^2}$ entries, and computing the value of each entry can be done in polynomial time as described  above.
\end{proof}

\section{Conclusion}
We have derived a polynomial-time algorithm for the {POSP} of MCS restricted to chordal graphs, resolving one open question left in the work of Scheffler~\cite{DBLP:conf/wg/Scheffler22}. To achieve this result, we propose the notion of layer structure which might be of independent interest.

We remark that despite that our algorithm is for the {\prob{POSP}} which is defined as a decision problem, the algorithm can be utilized to solve its optimization version: if an instance $(G, R)$ is determined as a {\yesins} by our algorithm, there exists $D(0, C(G), K)=1$ for some maximal clique~$K$ of~$G$. Then, we can use standard backtracking techniques to obtain a Prim ordering  of~$C(G)$ that starts with~$K$ and respects~$R$, and use Lemma~\ref{lem-mcs-prim-ordering-b} to obtain an MCS ordering of~$G$ extending~$R$ in polynomial time.

Additionally, graph search paradigms are usually studied for connected graphs but they can be trivially adapted for disconnected graphs by running the search algorithms on each connected component one after another. In this case, our result also holds if the input graph is disconnected. 
In particular, if~$G$ is disconnected in an instance $I=(G, R)$, we run the  algorithm presented in the proof of Theorem~\ref{thm-main} for each $(G', R')$ where~$G'$ is a connected component of~$G$ and~$R'=\rsto{R}{V(G')}$. If at least one of these instances is a {\noins},~$I$ is a {\noins}. Otherwise,~$I$ is a {\yesins} if and only if there  exist no distinct $x, y, z\in V(G)$ such that~$x$ and~$y$ are in the same connected component of~$G$,~$x$ and~$z$ are in different connected components of~$G$, and $\{(x, z), (z, y)\}\subseteq R$.

An important topic for future research is to improve the running time of our algorithm. Regarding this issue, by elementary analysis, one can verify easily that our algorithm runs in~$\bigo{n^4}$ time, where~$n$ is the the number of vertices in the input graph. On top of that, investigating Scheffler's~\cite{DBLP:journals/tcs/Scheffler22} primary question---whether there are important graph classes restricted to which POSP is {\nph} but at least one of the end-vertex problem and the {$\mathcal{F}$}-tree recognition problem is polynomial-time solvable---is another promising avenue for future research.

\section*{Acknowledgement}
This paper was supported by the National Natural Science Foundation of China under grant 62302060 and 62372066, Research Foundation of Education Bureau of Hunan Province under grant 21B0305, and Natural Science Foundation of Hunan Province of China under grant 2022JJ30620. 

The authors thank the anonymous reviewers of MFCS 2023 for their careful reading and instructive comments.

\end{document}